\pdfoutput=1
\documentclass[pdfa,a4paper,USenglish,autoref,thm-restate,cleveref]{lipics-v2021}
\usepackage{xspace}
\usepackage[linesnumbered,ruled,nofillcomment]{algorithm2e}
\usepackage{multirow}
\usepackage{hhline}
\usepackage{pifont}
\usepackage{multicol}
\usepackage{changepage}

\newenvironment{specification}[1][htb]
{
\begin{algorithm}[#1]}{\end{algorithm}}

\SetKwBlock{KwFunction}{function}{}
\SetKwBlock{KwUpon}{upon}{}
\SetKwBlock{KwOn}{on}{}
\SetKwBlock{KwAfter}{after}{}
\SetKwBlock{KwInitialize}{initialize}{}
\SetKwBlock{KwTimer}{Timer:}{}
\SetKwBlock{KwLeader}{Leader:}{}
\SetKwBlock{KwState}{initialize}{}
\SetKwBlock{KwRepeat}{repeat}{}
\SetKwBlock{KwPeriodically}{periodically}{}
\SetKwBlock{KwState}{state}{}
\SetKwBlock{KwAtomic}{atomically}{}
\SetKwBlock{KwConcurrently}{concurrently}{}
\SetCommentSty{textup} 
\SetKwComment{KwIttayComment}{(}{)} 
\SetKw{KwAnd}{and} 
\SetKw{KwOr}{or}
\SetKw{KwSetTimer}{setTimer}
\SetKwBlock{KwExpTimer}{on timer expiration}{}
\SetKwFor{ForAll}{forall}{}

\SetKwIF{If}{ElseIf}{Else}{if}{}{else if}{else}{end if}
\newcommand{\balance}{\ensuremath{\textit{bal}}\xspace}
\newcommand{\amount}{\ensuremath{\textit{amt}}\xspace}

\hideLIPIcs

\crefname{algocf}{alg.}{algs.}
\Crefname{algocf}{Alg.}{Algs.}

\crefformat{chapter}{\S#2#1#3}
\crefmultiformat{chapter}{\S\S#2#1#3}{ and~#2#1#3}{, #2#1#3}{, and~#2#1#3}

\crefformat{section}{\S#2#1#3}
\crefmultiformat{section}{\S\S#2#1#3}{ and~#2#1#3}{, #2#1#3}{, and~#2#1#3}

\title{On Payment Channels in Asynchronous Money Transfer Systems}

\author{Oded Naor}{Technion}{}{}{Oded Naor is grateful to the Azrieli Foundation for the award of an Azrieli Fellowship, and to the Technion Hiroshi Fujiwara Cyber-Security Research Center for providing a research grant.}
\author{Idit Keidar}{Technion}{}{}{}

\authorrunning{O. Naor and I. Keidar}
\Copyright{Oded Naor and Idit Keidar}

\ccsdesc[500]{Computing methodologies~Distributed algorithms}

\keywords{Blockchains, Asynchrony, Byzantine faults, Payment channels}

\nolinenumbers

\EventEditors{Christian Scheideler}
\EventNoEds{1}
\EventLongTitle{36th International Symposium on Distributed Computing (DISC 2022)}
\EventShortTitle{DISC 2022}
\EventAcronym{DISC}
\EventYear{2022}
\EventDate{October 25--27, 2022}
\EventLocation{Augusta, Georgia, USA}
\EventLogo{}
\SeriesVolume{246}
\ArticleNo{29}

\begin{document}
\sloppy
\maketitle
\vspace{0.5\baselineskip}
\enlargethispage{-0.5\baselineskip}
\begin{abstract}
Money transfer is an abstraction that realizes the core of cryptocurrencies.
It has been shown that, contrary to common belief, money transfer in the presence of Byzantine faults can be implemented in asynchronous networks and does not require consensus.
Nonetheless, existing implementations of money transfer still require a quadratic message complexity per payment, making attempts to scale hard.
In common blockchains, such as Bitcoin and Ethereum, this cost is mitigated by \emph{payment channels} implemented as a second layer on top of the blockchain allowing to make many off-chain payments between two users who share a channel.
Such channels require only on-chain transactions for channel opening and closing, while the intermediate payments are done off-chain with constant message complexity.
But payment channels in-use today require synchrony; therefore, they are inadequate for asynchronous money transfer systems.

In this paper, we provide a series of possibility and impossibility results for payment channels in asynchronous money transfer systems. 
We first prove a quadratic lower bound on the message complexity of on-chain transfers.
Then, we explore two types of payment channels, unidirectional and bidirectional.
We define them as shared memory abstractions and prove that in certain cases they can be implemented as a second layer on top of an asynchronous money transfer system whereas in other cases it is impossible.
\end{abstract}

\section{Introduction}

\begin{table*}[t] 
\caption{Summary of the results.}
\label{table:intro}
\centering
\setlength\doublerulesep{0.5pt}
\resizebox{\textwidth}{!}{\begin{tabular}{| l | l  | l  c |}
\hline
\textbf{Abstraction} & \textbf{Operations} & \multicolumn{1}{c|}{\textbf{Upper bound}} & \textbf{Lower bound} \\
&&\multicolumn{1}{c|}{\textbf{message complexity}}& \textbf{message complexity}\\
\hhline{====}
Asset transfer & transfer & \multicolumn{1}{l|}{$O(n^2)$  \multirow{2}{*}{   \cite{guerraoui2019consensus,auvolat2020money} }} & \multirow{2}{*}{$\Omega(n^2)$ [\Cref{thm:assetTransferMsgComplexity}]} \\ 

&  read & \multicolumn{1}{l|}{$O(1)$  }&  \\
\hhline{====}
Bidirectional payment channel & open & \multicolumn{ 2}{c | }{\multirow{3}{*}{Impossible in asynchronous networks [\Cref{cor:bidirectionalAsync}] }}    \\
& transfer &  &  \\
& close & & \\
\hhline{====}
Unidirectional payment channel &open & \multicolumn{ 2}{c | }{\multirow{4}{*}{Impossible in asynchronous networks [\Cref{cor:unidirectionalAsync}] } }\\
with source close &transfer&& \\
& source close && \\
&target close && \\
\hline

Unidirectional payment channel & open &  \multicolumn{1}{l|}{$O(n^2)$  \multirow{3}{*}{   [Alg.~\ref{alg:uniPaymentChannelImplementation}] } }  & \multirow{3}{*}{$\Omega(n^2)$ [\Cref{lem:uniChannelMessageComplexity}]}	\\ 
& transfer & \multicolumn{1}{l|}{$O(1)$ }&  \\ 
& target close &\multicolumn{1}{l|}{$O(n^2)$ }&    \\
\hline
\end{tabular}}
\end{table*}

The rise of \emph{cryptocurrencies}, such as Bitcoin~\cite{nakamoto2008bitcoin}, Ethereum~\cite{wood2014ethereum},  Ripple~\cite{armknecht2015ripple}, and many more, has revolutionized the possibility of using decentralized money systems.

In 2019, Guerraoui et al.~\cite{guerraoui2019consensus} defined the abstraction of \emph{asset transfer} or \emph{money transfer} capturing the original motivation of Bitcoin.
This  abstraction is based on a set of known users owning accounts, each account has some initial money, and the users can transfer money between the accounts.
It is well-known that deterministic consensus cannot be solved in an asynchronous network~\cite{fischer1985FLP}, meaning that blockchains that rely on consensus require synchrony to work properly.
Nonetheless, Guerraoui et al. showed that the asset transfer problem is weaker than consensus~\cite{lamport1982byzantine,dwork1988consensus}, i.e., in the Byzantine message-passing model the problem can be solved in an asynchronous network.
They provide a concrete implementation of the abstraction in this model  using an asynchronous broadcast service.

Yet, in this solution, each payment requires a  message complexity of  $O(n^2)$, where $n$ is the number of processes in the system.
If the number of processes grows, this per-payment quadratic message complexity can pose a real challenge in scaling the asset transfer network.
In fact, \emph{scalability} is one of the major limiting factors of blockchains and consensus protocols, and extensive research was done to reduce the message complexity~\cite{naor2020expected,eyal2016bitcoin,keidar2021all,yin2019hotstuff,cohen2020coincidence} in various settings.

A promising approach to scale blockchain payments is by using \emph{payment channels}~\cite{poon2016lightning,lind2019teechain,raiden2020raidenImp} as a second off-chain layer on top of the blockchain.
A payment channel can be \emph{opened} between two blockchain account owners via an on-chain deposit made to fund the channel, after which the two users can \emph{transfer} payments on the channel itself off the blockchain.
At any time, one of the users can decide to \emph{close} the channel, after which  the current balance in the channel of each of the users is transferred to their on-chain accounts.
In this scheme, opening or closing a channel requires a blockchain transaction, which incurs a large message complexity, but all the intermediate payments on the channel require message exchange only between the channel users, freeing the blockchain from these payments and messages.

Payment channels have been actively deployed in central blockchains.
For example, Bitcoin's \emph{Lightning Network~(LN)}~\cite{poon2016lightning} is a highly used payment network, with active channels holding a Bitcoin amount equivalent to hundreds of millions of dollars~\cite{1ml}.

One of the downsides of using payment channels such as LN is that the payment channels themselves rely on network synchrony.
For example,  in LN, suppose Alice and Bob have an open payment channel between them and Alice acts maliciously and tries to steal money by closing the channel at a stale state.
LN provides a way for Bob to penalize Alice and confiscate all the money in the channel.
But to do so, Bob has to detect Alice's misbehavior on-chain and act within a predetermined time frame, making this method inappropriate with asynchronous users.

In this paper, we explore the possibility of implementing payment channels in asynchronous asset transfer systems.
The results of the paper and prior art are summarized in~\Cref{table:intro}.
We study four abstractions: asset transfer and different types of payment channels.

First, we prove a lower bound on the complexity of asset transfer without channels.
We show that no matter what implementation is provided for the asset transfer abstraction, it still requires $\Omega(n^2)$ messages for a single payment to be made by one party and observed by others.
This means that the upper bound is tight for the algorithm provided in~\cite{guerraoui2019consensus}, in which transferring money has $O(n^2)$ message complexity and reading the balance of an account costs $O(1)$.
This fundamental result shows that while the synchrony requirement can be relaxed for asset transfer, each payment still requires a rather large number of messages.
This means that second-layer solutions such as payment channels are required for scalability.

Next, we consider payment channel abstractions with operations for opening a channel, transferring money in it, and closing the channel.
We first consider a \emph{bidirectional payment channel} as a second layer atop an asset transfer system, where each side of the channel can make payments to the other.
This is similar to LN~\cite{poon2016lightning}, Teechain~\cite{lind2019teechain}, Sprites~\cite{miller2017sprites}, and more payment channel proposals.
Once we formalize the problem, it is easy to show that synchrony is required, and therefore a bidirectional channel cannot be implemented on top of an asynchronous asset transfer system.

We next look at \emph{unidirectional payment channels}, in which only one user, the channel \emph{source}, can make payments to the other user, the \emph{target}.
We differentiate between two types of unidirectional payment channels: 
If we allow both the source and the target to close the channel (the source and target close operations, respectively), we again show that synchrony is required.
Indeed, a previous design of similar channels~\cite{spilman_2013} still requires synchrony.

On the other hand, if we allow only the target to close the channel, we provide a concrete implementation that works in an asynchronous network.
In this implementation, the opening and closure of the channel require a payment using the asset transfer system, incurring $O(n^2)$ message complexity, whereas every payment on the channel itself requires a single message from the source user to the target user.
We note that once the channel supports a target close the target can claim its transferred funds in the channel, which is not the case when only the source can close the channel.

Finally, we outline an extension of payment channels to payment chains, in which payments are made across multiple channels atomically.
Like their 2-party counterparts, a chain payment over unidirectional channels with only target close can be implemented using a technique used in LN.
As for other channel types, $k$-hop chains are equivalent to $k$-process consensus, i.e., have a consensus number~\cite{herlihy1991wait} of $k$.

To conclude, our contributions in this paper are as follows:
\begin{itemize}
\item We prove a quadratic message complexity lower bound for asynchronous asset transfer systems.
\item We explore payment channels as a key scaling solution for asynchronous asset transfer systems and provide a series of impossibility results for different channel types.
\item  We provide a concrete possibility result and an implementation for an asynchronous payment channel.
\end{itemize}

\subparagraph*{Structure.}
The rest of the paper is structured as follows: \Cref{sec:model} describes the model and preliminaries; \Cref{sec:assetTransfer} details the asset transfer abstraction and proves a lower bound on message complexity; \Cref{sec:biChannel} discusses bidirectional payment channels and \Cref{sec:uniChannel} discusses unidirectional channels;  \Cref{sec:chain} extends the discussion to chain payments; \Cref{sec:relatedWork} discusses related work; and finally, \Cref{sec:conclusion} concludes the paper.

\section{Model and preliminaries}
\label{sec:model}
We study a message-passing distributed system that consists of a set $\Pi = \{p_1, \ldots, p_n\}$ of $n$ processes.
The processes can interact among themselves by sending messages.
An adversary can \emph{corrupt} up to $f < n/3$ processes, where $f  \in \Theta(n)$.
If not mentioned explicitly, we assume a corrupt process is \emph{Byzantine}, i.e., it can deviate from the prescribed algorithm and act arbitrarily.
Any non-corrupt process is \emph{correct}.
A \emph{crash-fail} fault is when a process stops participating in the algorithm.
Every two  processes share an asynchronous reliable link between them, such that if one correct process sends a message to another correct process, it eventually arrives, and the target can ascertain its source.

We assume the existence of a \emph{Public Key Infrastructure (PKI)}, whereby processes that know a private key can use it to sign messages such that all other processes can verify the signature. 
The adversary cannot forge a signature if the private key is owned by a correct process.
We assume each private key is owned by one process.
We further assume \emph{multisignatures}~\cite{bellare2007identity}, whereby in order to produce a valid signature matching a constant-sized public key, more than one private key is used to sign the message.
Note that signing can be sequential.

We study algorithms in the message-passing model that implement  abstractions that are defined as \emph{shared-memory objects}.
A shared memory object has a set of \emph{operations}, and processes access the object via these operations.
Each operation starts with an \emph{invocation} event by a process and ends with a subsequent \emph{response} event.
Invocations and responses are discrete events.

An \emph{implementation} or an \emph{algorithm} $\pi$ of a shared-memory object abstraction is a distributed protocol that defines the behaviors of processes as deterministic state machines, where state transitions are associated with \emph{actions}: sending or receiving messages, and operation invocations or responses.
A \emph{global state} of the system is a mapping to states from systems components, i.e., processes and links.
An \emph{initial global state} is when all processes are in initial states and there are no messages on the links between the processes.
A \emph{run} or an \emph{execution} of an implementation is an alternating series of global states and actions, beginning with some initial global state, such that state transitions occur according to $\pi$.
We assume that the first action of each process is an invocation of an operation and that it does not invoke another operation before receiving a response for its last invoked operation.

Each execution creates a \emph{history} $H$ that consists of a sequence of matching invocations and responses, each with the assigned process that invoked the operation and the matching responses.
A \emph{sub-history $H'$ of $H$} is a subset of $H$'s events.
Let $H|_p$ denote the sub-history of $H$ with process $p$'s events.

A history defines a \emph{partial ordering}: \emph{operation $\textit{op}_1$ precedes $\textit{op}_2$ in history $H$}, labeled $\textit{op}_1 \prec_H \textit{op}_2$, if $\textit{op}_1$'s response event happens before $\textit{op}_2$'s invocation event in $H$.
History $H$ is \emph{sequential} if  each invocation, except perhaps the last, is immediately followed by a matching response.
An operation is \emph{pending} in history $H$ if it has an invocation event in $H$ but does not have a matching response. 
A history $H'$ is a \emph{completion of history $H$} if it is identical to $H$ except for removing zero or more pending operations in $H$ and by adding matching responses for the remaining ones.
A shared-memory abstraction is usually defined in terms of a \emph{sequential specification}.
A \emph{legal} sequential history is a sequential history that preserves the sequential specification, i.e98., the sequential specification is the set of all legal histories.

The correctness criteria we consider is \emph{Byzantine sequential consistency (BSC)}.
This allows to extend sequential consistency~\cite{attiya1994sequential} to runs with Byzantine processes and not only crash-fail and is an adaptation of the definition of Byzantine linearizability~\cite{cohen2021ByzantineLinea}.

First, we formally define a sequentially consistent history for runs with crash-fail faults.
\begin{definition} [sequentially consistent history] \label{def:sequentialConsistency}
Let $E$ be an execution of an algorithm, and $H$ its matching history.
Then $H$ is sequentially consistent if there exists a completion $\widetilde{H}$ of $H$ and a legal sequential history $S$ such that for every process $p$, $S|_p = \widetilde{H}|_p$.
\end{definition}
An algorithm is sequentially consistent if all its histories are sequentially consistent.

For Byzantine sequential consistency (BSC), let $H|_c$ denote the sub-history of $H$ with all of the operations of correct processes.
We say a history is BSC if $H|_c$ can be augmented with operations of Byzantine processes such that the completed history is sequentially consistent.
Formally:
\begin{definition} [Byzantine sequential consistency (BSC)]
A history $H$ is BSC if there exists a history $H'$ such that $H'|_c = H|_c$, and $H'$ is sequentially consistent.
\end{definition}
Similar to sequential consistency, an algorithm is BSC if all its histories are BSC.
We choose BSC as the correctness criterion and not Byzantine linearizability because it simplifies the implementations we provide below.
We explain in \Cref{sec:uniChanelImplementation:UpperBound} how to change the implementation that we provide to satisfy Byzantine linearizability.
The difference between sequential consistency and from linearizability~\cite{herlihy1990linearizability} is that linearizability also preserves real-time order, i.e., for any operations $op_1, op_2$ s.t. $op_1 \prec_{\widetilde{H}} op_2$, then $op_1 \prec_S op_2$.

\section{Asset transfer}
\label{sec:assetTransfer}
\subsection{Asset transfer abstraction}
Let $A$ be an asset transfer abstraction, which is based on the one defined in~\cite{guerraoui2019consensus}.
$A$ holds a set of accounts.
Each account $a \in A$ is defined by some public key, and there is a mapping $\textit{owner}(a) \colon A \mapsto \Pi$, that matches for each account $a$ the process that can produce a signature corresponding to the public key associated with $a$.
In case of an account $b$ associated with a multisignature, $\textit{owner}(b)$ is the set of processes whose private keys can produce a signature matching $b$'s public key.
The state of each account $a$ is of the form $A(a) \in \mathbb{R}_{\geq 0}$ and represents the balance of the account~$a$.
Each account initially holds its initial balance.

$A$ has two operations:
The first, $A.\textit{read}(a)$, returns the balance of account $a$, i.e., it returns $A(a)$ and can be called by any process.
The second is $A.\textit{transfer}(a,[(b_1,\amount_1),\ldots,(b_k,\amount_k)])$, which, for every $1 \leq i \leq k$,  transfers from account $a$'s balance $\amount_i$ and deposits it in $b_i$. 
This call succeeds, and returns \textit{success} if it is called by $\textit{owner}(a)$ and if the account has enough balance to make the transfer, i.e., $A(a) \geq \sum_{i=1}^k \amount_i$.
Otherwise, it returns \textit{fail} and does nothing.

The \textit{transfer} operation is an extension of the one defined in~\cite{guerraoui2019consensus} in that it allows to transfer money from one account to multiple accounts, whereas the original work only allows transferring money from one account to another each time.  

In this paper, we consider implementations of the asset transfer abstraction that are BSC.
We note that the message-passing asset transfer implementation in~\cite{guerraoui2019consensus} is based on a reliable broadcast that preserves source order.
Their correctness criteria is neither Byzantine linearizability nor BSC, but it ensures that  for every transfer operation, there exists a time $t$ such that if a correct process invokes the read operation after $t$, then it observes the changes made by the transfer.
This is a property we use throughout the proofs which are detailed below.

\subsection{Message complexity of asset transfer}
We begin by proving a quadratic message complexity lower bound on any asset transfer implementation in \Cref{sec:assetTransfer:msgComplexity:lowerBound}, and 
\Cref{sec:assetTransfer:msgComplexity:UpperBound} shows this lower bound is tight by discussing a concrete implementation of an asset transfer that has a quadratic message complexity for a transfer operation and a constant message complexity for a read operation.

\subsubsection{Lower bound}
\label{sec:assetTransfer:msgComplexity:lowerBound}
We show that  if $f \in \Theta(n)$, then there is a quadratic message complexity lower bound in runs in which money is transferred to some account $b$, and multiple processes read the balance of $b$.
The proof we use for the lower bound follows the technique used in Dolev-Reischuk's lower bound for Byzantine Broadcast~\cite{dolev1985bounds}.
Formally, we prove the following theorem:
\begin{theorem} \label{thm:assetTransferMsgComplexity}
Consider an algorithm that implements the asset transfer abstraction.
Then there exists a run with a single transfer invocation and multiple read invocations in which the correct processes send at least $(f/2)^2$ messages.
\end{theorem}

\begin{proof}
Let $\pi$ be an algorithm that implements asset transfer, and assume by contradiction that in all its runs with a single transfer correct processes send less than $(f/2)^2$ messages.
We look at all executions of $\pi$ with two accounts $a, b \in A$ s.t. $\textit{owner}(a) = p$ for some process $p \in \Pi$, and initially $A(a) = 1, A(b) = 0$.

Consider first an execution $\sigma_0$ in which the adversary, denoted $\textit{adv}_0$, corrupts a set $V$ of processes, not including $p$, such that $|V| = \lceil f/2 \rceil$.
Denote the set of remaining correct processes as $U$.
In $\sigma_0$, process $p$ calls $A.\textit{transfer}(a, [(b, 1)])$.
By the correctness definition of $A$, and since $p$ is correct, there exists a time $t_0$ during the run after which any correct process that invokes $A.\textit{read}(b)$  returns 1. 

The adversary $\textit{adv}_0$ causes the corrupt processes in $V$ to simulate the behavior of correct processes that call $A.\textit{read}(b)$ after $t_0$, and follow the algorithm except for the following changes: they ignore the first $f/2$ messages they receive from processes in $U$, and they do not send any message to other processes in $V$.
Note that while $t_0$ is not known to the processes, we construct the runs from the perspective of a global observer and may invoke read after $t_0$. 

Because correct processes send, in total, less than $(f/2)^2$ messages and corrupt processes do not send messages to other processes in $V$, then the processes in $V$ together receive less than $(f/2)^2$ messages. 
Thus, by the pigeonhole principle, there exists at least one process $q \in V$ that receives less than $f/2$ messages.
Denote the set of processes that send messages to $q$ as $U'$, and denote $U'' = U \setminus U'$.
Note that $U'$ may include process $p$, and that $|U'| < f/2$.

Next, we construct a  run $\sigma_1$ with an adversary $\textit{adv}_1$ that are the same as $\sigma_0$ and $\textit{adv}_0$, respectively, except for the following changes: $\textit{adv}_1$ corrupts all the processes in $V \setminus \{q\}$, and all the processes in $U'$.
Since $|U'| < f/2$ and $|V| \leq \lceil f/2 \rceil$, the adversary $\textit{adv}_1$ corrupts at most $f$ processes in $\sigma_1$.
$\textit{adv}_1$ prevents the corrupt processes from sending any message to $q$, but causes them to behave correctly towards all other correct processes in $U''$.

By definition, the behavior of the corrupt processes in $\sigma_1$, i.e., the processes in $U' \cup (V \setminus \{q\})$, towards the correct processes in $U''$ is the same as in $\sigma_0$.
Since process $q$ simulates a correct process that ignores the first $f/2$ messages in $\sigma_0$, its behavior towards the processes in $U$ is identical in both runs as well.
Thus, runs $\sigma_0$ and $\sigma_1$ are indistinguishable for the correct processes in $U''$, ensuring that they behave the same.
Since process $q$ acts in $\sigma_0$ like a correct process that does not receive any message, both runs are indistinguishable to it as well.

Nonetheless, process $q$ still has to return a value for its $A.\textit{read}(b)$ call.
Denote the time when the call returns as $t_1$.
If it returns a value different from $1$, then we conclude the proof, since it is a violation of the read call specification.
Otherwise, we construct a run $\sigma_2$ with an adversary $\textit{adv}_2$ that are the same as $\sigma_1$ and $\textit{adv}_1$, respectively, except that there is no transfer invocation and all messages to $q$ are delayed until after $t_1$.
For process $q$, runs $\sigma_1$ and $\sigma_2$ are indistinguishable until $t_1$, therefore it returns $1$ for the $A.\textit{read}(b)$ call, violating the read call specification which should return $0$, concluding the proof.
\end{proof}

We proved that the lower bound for the message complexity of an asset transfer object is $\Omega(n^2)$, assuming $f \in \Theta(n)$. 

\subsubsection{Upper bound}
\label{sec:assetTransfer:upperBound} \label{sec:assetTransfer:msgComplexity:UpperBound}
In~\cite{guerraoui2019consensus}, an implementation in the message-passing model for the asset transfer abstraction is provided.
It uses a broadcast service defined in~\cite{malkhi1997high} that tolerates up to $f < n/3$ Byzantine failures.
This broadcast has all the guarantees of reliable broadcast~\cite{bracha1987asynchronous} (integrity, agreement, and validity), and also preserves source order, i.e., any two correct processes $p_1$ and $p_2$ that deliver  messages $m$ and $m'$ broadcast from the same process $p_3$, do so in the same order.
The read operation is computed locally, and the transfer operation consists of a broadcast of a single message.

Several protocols can be used to implement such a source-order broadcast service, including a protocol in~\cite{malkhi1997high}.
Bracha's reliable broadcast~\cite{bracha1987asynchronous} can also be used to implement such a service if each correct process adds a sequence number to each message it broadcasts, and each correct process delivers messages from the same process in the order of the sequence numbers.
These protocols have a message complexity of $O(n^2)$ per broadcast, proving that the lower bound message complexity we prove above is tight.

Note that our definition for the \textit{transfer} call of the asset transfer abstraction allows transferring in each invocation money from one account to multiple accounts, while in~\cite{guerraoui2019consensus} the transfer call allows a transfer to a single account for each invocation.
The implementation in~\cite{guerraoui2019consensus} can easily be adjusted to support this change by including in each broadcast message the multiple accounts to which money is transferred. {}
\section{Bidirectional payment channel}
\label{sec:biChannel}
We seek to analyze if payment channels that are used in common blockchains as a second layer can also be used similarly on top of an asynchronous asset transfer system.
We discuss bidirectional payment channels, in which a channel is opened between two processes by making a transfer on the asset transfer system.
After the channel is opened, both processes can make bidirectional payments on the same channel.
Either process can close the channel at any time, after which their accounts in the asset transfer system reflect the state of the channel.
This abstraction is similar to payment channels in the Lightning Network~\cite{poon2016lightning} in Bitcoin~\cite{nakamoto2008bitcoin} and Raiden~\cite{raiden2020raidenImp} in Ethereum~\cite{wood2014ethereum}.
We compare our payment channel abstractions to the currently available implementations in the related work in \Cref{sec:relatedWork}.

First, we formally define this abstraction, and then provide an impossibility result, proving it cannot be implemented in asynchronous networks.

\subsection{Definition}
\begin{specification}[t]  \smaller

\caption{\textbf{Bidirectional payment channel abstraction.} Operations for process~$p$.}
\SetAlgoNoEnd
\label{alg:bidirectionalChannelAPI}
\DontPrintSemicolon
\SetInd{0.4em}{0.4em}

\let\oldnl\nl
\newcommand{\nonl}{
\renewcommand{\nl}{\let\nl\oldnl}}

\SetKwBlock{sharedObjects}{Shared Objects:}{}
\SetKwBlock{localVariables}{Local variables:}{}
\SetKwProg{myproc}{Procedure}{:}{}
\SetKwProg{myfunc}{Function}{:}{}

\nonl \sharedObjects{
\nonl	$A$ - asset transfer object \;
\nonl	$BC$ - Bidirectional payment channel object
}

\begin{multicols}{2}

\myproc{$BC.\textit{transfer}((a,b),\amount)$} {
\If{$BC(a,b) = \bot $} {
\Return}
$(\balance_a, \balance_b) \gets  BC(a,b)$ \;
\If{$\textit{owner}(a) = p \wedge \balance_a \geq \amount$}  {
$\textit{execute\_payment}((a,b), -\amount, \amount) $\;

}
\If{$\textit{owner}(b) = p \wedge \balance_b \geq \amount$}  {
$\textit{execute\_payment}((a,b), \amount, -\amount) $\;
}
}
\BlankLine

\myfunc{$\textit{execute\_payment}((a,b), \amount_a, \amount_b) $}
{
$(\balance_a, \balance_b) \gets  BC(a,b)$ \;
$\textit{new\_bal}_a \gets \balance_a + \amount_a$ \;
$\textit{new\_bal}_b \gets \balance_b + \amount_b$ \;
$ BC(a,b) \gets (\textit{new\_bal}_a, \textit{new\_bal}_b)$ \;
}
\end{multicols}
\BlankLine \BlankLine
\begin{multicols}{2}	
\myproc{$BC.\textit{close}((a,b), \balance)$} {
\If{$B(a,b) = \bot$} {
\Return \textit{fail}
}

$(\textit{curr\_bal}_a, \textit{curr\_bal}_b) \gets  BC(a,b)$ \;
$\textit{other\_bal} = \textit{curr\_bal}_a+\textit{curr\_bal}_b-\balance$ \;

\If{$\textit{owner}(a)  = p $} {
\If{$\balance \neq \textit{curr\_bal}_a $}{
\Return \textit{fail}
}
\Return $\textit{execute\_close}((a,b), \balance, \textit{\textit{other\_bal}})$
}
\If{$\textit{owner}(b)  = p $} {
\If{$\balance \neq \textit{curr\_bal}_b $}{
\Return \textit{fail}
}
\Return $\textit{execute\_close}((a,b), \textit{other\_bal}, \balance)$
}
\Return \textit{fail}
}

\myfunc{$\textit{execute\_close}((a,b), \amount_a, \amount_b) $} {
$A(a) \gets A(a) + \amount_a$  \;
$A(b) \gets A(b) + \amount_b$ \;
$BC(a,b) \gets \bot$ \;
\Return \textit{success}
}

\end{multicols}
\BlankLine
\end{specification}
\newcommand{\biChannel}{\ensuremath{BC}}

We define a bidirectional payment channel abstraction as a shared memory object $\biChannel$.
The formal definition is in  Specification~\ref{alg:bidirectionalChannelAPI}.
$\biChannel$ is defined based on the existence of an asset transfer object $A$.
A channel in $\biChannel$ is of the form $(a, b)$, where $a, b$ are accounts in $A$.

The state of a payment channel $\biChannel(a,b)$ is ${\{\mathbb{R}_{\geq 0} \times \mathbb{R}_{\geq 0} \} \cup \{\bot \}}$.
The channel $\biChannel(a,b)$ can be \emph{open}, and then $\biChannel(a,b) = (\balance_a, \balance_b)$, which represents the balances $\balance_a, \balance_b$  of accounts $a,b$ in the channel $(a,b)$, respectively.
If the channel is \emph{closed}, then its state is $\biChannel(a,b) = \bot$.

The set of operations is the following:
\begin{itemize}
\item $\textit{open}$.
We do not provide a detailed specification for this call, as we do not require the full specification to prove that there is no implementation for a bidirectional payment channel in the asynchronous message-passing model. 
Instead, we assume that all channels are open at the beginning of the run with some initial balances.

\item $\textit{transfer}((a,b),\amount)$.
A payment in channel $(a,b)$ is possible if the channel is open, if the caller process is a valid owner of either $a$ or $b$, and the caller's balance in the channel is enough to make the payment. Otherwise, it does nothing.
After the call ends, the state $\biChannel(a,b)$ is changed to reflect the payment.

\item $\textit{close}((a,b), \balance)$.
This call can be invoked if the caller process is a valid owner of either $a$ or $b$.
The close of the channel is successful if the process that invokes the call does not try to close it with balance \balance that is not the amount it has in the channel.
Otherwise, the call fails and does nothing.
The call transfers to accounts $a,b$ their balances from the channel, and then changes its status to $\bot$.
This call returns \textit{success} or \textit{fail} to indicate the outcome of the call.
\end{itemize}
In this paper, we consider sequentially consistent bidirectional channels.

\subsection{Impossibility of a bidirectional payment channel object}

\begin{algorithm*}[t]  \smaller
\caption{\textbf{Wait-free implementation of consensus among 2 processes using a bidirectional payment channel.} Operations for processes $p_1 = \textit{owner}(a), p_2 = \textit{owner}(b)$.}
\SetAlgoNoEnd
\label{alg:reductionConsensuswithBidirectionalChannel1}
\DontPrintSemicolon
\SetInd{0.4em}{0.4em}

\let\oldnl\nl
\newcommand{\nonl}{
\renewcommand{\nl}{\let\nl\oldnl}}

\SetKwBlock{sharedObjects}{Shared Objects:}{}
\SetKwProg{myproc}{Procedure}{:}{}
\SetKwProg{myfunc}{Function}{:}{}

\nonl \sharedObjects{
\nonl	$A$ - asset transfer object, initially two accounts $a, b \in A$ s.t. $A(a) = A(b) = 0$ \;
\nonl	$\biChannel$ - Bidirectional payment channel object, initially $\biChannel(a,b) = (1,1)$ \;
\nonl	$R_1, R_2$ - shared registers with read write calls, initially $R_1 = R_2 = \bot$ \;
}

\BlankLine 
\tcp*[h]{We assume  that at the beginning of the run there exists an open payment channel $(a,b)$ s.t. $\biChannel(a,b) = (1,1)$} \;	
\BlankLine
\begin{multicols}{2}

\tcp*[h]{\underline{Algorithm for process $p_1$}:} \;	
\myproc{$\textit{propose}(v)$} {
$R_1.\textit{write}(v)$ \label{alg:reductionConsensuswithBidirectionalChannel:r1write}\;
$\biChannel.\textit{transfer}((a,b),1) $\;
$\biChannel.\textit{close}((a,b), 0)$ \label{alg:reductionConsensuswithBidirectionalChannel1:p1close}\;
\Return $\textit{make\_decision}()$ \;
}

\BlankLine
\tcp*[h]{\underline{Algorithm for process $p_2$:}} \;	
\myproc{$\textit{propose}(v)$} {
$R_2 .\textit{write}(v)$  \label{alg:reductionConsensuswithBidirectionalChannel:r2write}\;
$\biChannel.\textit{close}((a,b), 1)$ \label{alg:reductionConsensuswithBidirectionalChannel1:p2close} \;
\Return $\textit{make\_decision}()$ \;
}
\break
\BlankLine
\tcp*[h]{\underline{Algorithm for processes $p_1$ and $p_2$:}} \;	
\myproc{$\textit{make\_decision}()$ } {
\textbf{wait} until $A.\textit{read}(b) \neq 0$  \label{alg:reductionConsensuswithBidirectionalChannel:ARead}\;
\uIf{$A.\textit{read}(b) = 2$} {
\Return $R_1.\textit{read}()$ \;
}
\Else {
\Return $R_2.\textit{read}()$ \;
}
}
\end{multicols}
\BlankLine
\end{algorithm*}

We show that implementing a sequentially consistent bidirectional payment channel in the message-passing model requires synchrony.
To this end, we solve wait-free consensus among $2$ processes with shared registers and an instance of $\biChannel$.
The consensus abstraction has one call, $\textit{propose}(v)$, which is called with some proposal $v$, and returns a value.
The returned value for any process making the call has to be an input of the call from one of the processes, and it has to be the same value for all invocations, regardless of the caller process.
We assume in this proof \emph{crash-fail} faults, i.e., a process corrupted by the adversary stops participating in the protocol but does not deviate from it.
Since this is an impossibility result and crash-fail faults are weaker than Byzantine faults, it also applies to runs with Byzantine processes.

\enlargethispage{\baselineskip}
\begin{lemma} \label{lem:consensusNumBiChannel}
Consensus has a wait-free implementation for $2$ processes in the read-write shared memory model with an instance of a bidirectional payment channel shared-memory object and shared registers.
\end{lemma}
\begin{proof}
The algorithm for solving consensus among two processes using a bidirectional payment channel object is detailed in Alg.~\ref{alg:reductionConsensuswithBidirectionalChannel1}. 
We assume that there are two processes $p_1, p_2$ with ownership of accounts $a,b$, respectively, and an open payment channel $(a,b)$ at the beginning of the run with balances $\biChannel(a,b) = (1,1)$.

Before either of the processes invokes an operation on the payment channel, they write their proposal $v$ in a shared register (Lines~\ref{alg:reductionConsensuswithBidirectionalChannel:r1write} and~\ref{alg:reductionConsensuswithBidirectionalChannel:r2write}).
Then, $p_1$ attempts to make a payment on the channel and then close it, and $p_2$ tries to close the channel without making or accepting any payment.

Because $BC$ is sequentially consistent, the algorithm ensures that eventually after the channel is closed either the payment from $a$ to $b$ on the channel succeeds or not, and the balance in $b$'s account reflects it, i.e., there exists a time $t$ after which invoking $A.\textit{read}(b)$ returns either $1$ or $2$.

If the read call in \Cref{alg:reductionConsensuswithBidirectionalChannel:ARead} returns $2$, then the channel was closed by $p_1$ after it successfully made the payment on the channel.
Since before $p_1$ makes the payment on the channel it writes its proposal to register $R_1$, then its value is already written by the time the channel is closed, and it is returned by the propose call.
If the return value of the read call is~$1$, then process $p_2$ closed the channel.
Since $p_2$ closes the channel after it writes its proposal in $R_2$, then its proposed value is returned.

In either case, when the channel is closed, there is already a proposal written in either $R_1$ or $R_2$, i.e., the returned value is an input to the propose call by either process, and both processes return the same value.
\end{proof}

Based on the above theorem and FLP~\cite{fischer1985FLP}, we get the following result:
\begin{theorem} \label{cor:bidirectionalAsync}
There does not exist an implementation of the bidirectional payment channel abstraction in the asynchronous message-passing model.
\end{theorem}
\section{Unidirectional payment channel}
\label{sec:uniChannel}
\enlargethispage{\baselineskip}
After proving that a bidirectional payment channel cannot be implemented in asynchronous networks, we explore another type of payment channel, \emph{unidirectional}.
The main difference from bidirectional channels is that unidirectional channels are asymmetric.
There is only one user, the source, who can open and transfer money in the channel.
We show in which cases unidirectional payment channels can be implemented in an asynchronous message-passing network and in which cases they cannot.
We begin by formally defining the unidirectional payment channel abstraction.

\subsection{Definition}

We define a unidirectional payment channel abstraction as a shared memory object $B$.
The formal definition is in Specification~\ref{alg:uniPaymentChannelOps}. 
$B$ is defined based on the existence of an asset transfer object $A$.
A payment channel in $B$ is of the form $(a, b)$ where $a, b$ are accounts in $A$.

\begin{specification}[t]  \smaller
\caption{\textbf{Unidirectional payment channel abstraction}. Operations for process $p$.}
\SetAlgoNoEnd
\label{alg:uniPaymentChannelOps}
\DontPrintSemicolon
\SetInd{0.4em}{0.4em}

\let\oldnl\nl
\newcommand{\nonl}{
\renewcommand{\nl}{\let\nl\oldnl}}

\SetKwBlock{sharedObjects}{Shared Objects:}{}
\SetKwProg{myproc}{Procedure}{:}{}
\SetKwProg{myfunc}{Function}{:}{}

\nonl \sharedObjects{
\nonl	$A$ - asset transfer object, with initial accounts \;
\nonl 	$B$ - unidirectional payment channel object \;
}

\begin{multicols}{2}
\BlankLine \BlankLine
\myproc{$B.\textit{open}((a,b),\amount)$} {
\If{$p \neq \textit{owner}(a) \vee B(a,b) \neq \bot \vee A(a) < \amount$ } {
\Return \textit{fail}
}
$A(a) \gets A(a) - \amount$ \;
$B(a,b) \gets (\amount, 0)$ \;
\Return \textit{success}
}

\BlankLine
\myproc{$B.\textit{transfer}((a,b),\amount)$} {
\If{$p \neq \textit{owner}(a) \vee B(a,b) = \bot$} {
\Return
}
$(\balance_a, \balance_b) \gets  B(a,b)$ \;
\If{$\balance_a < \amount$} {
\Return 
}
$\textit{new\_bal}_a \gets \balance_a - \amount$ \;
$\textit{new\_bal}_b \gets \balance_b + \amount$ \;
$B(a,b) \gets  (\textit{new\_bal}_a, \textit{new\_bal}_b) $ \;
}

\BlankLine
\myproc{$B.\textit{target\_close}((a,b),\balance_b))$} {
\If{$p \neq \textit{owner}(b)  \vee  B(a,b) = \bot$}  {
\Return  \textit{fail}
}
$(\textit{curr\_bal}_a, \textit{curr\_bal}_b) \gets  B(a,b)$ \;
\If{$\balance_b \neq \textit{curr\_bal}_b$ }{
\Return  \textit{fail}
}
$\balance_a \gets \textit{curr\_bal}_a +  \textit{curr\_bal}_b - \balance_b$ \;
\Return $\textit{execute\_close}((a,b), (\balance_a, \balance_b))$ \;
}

\break
\myproc{$B.\textit{source\_close}((a,b), \balance_a)$} {
\If{$p \neq \textit{owner}(a)  \vee  B(a,b) = \bot$}  {
\Return  \textit{fail}
}
$(\textit{curr\_bal}_a, \textit{curr\_bal}_b) \gets  B(a,b)$ \;
\If{$\balance_a \neq \textit{curr\_bal}_a $ }{
\Return  \textit{fail}
}
$\balance_b \gets \textit{curr\_bal}_a +  \textit{curr\_bal}_b - \balance_a$ \;
\Return $\textit{execute\_close}((a,b), (\balance_a, \balance_b))$ \;
}

\BlankLine
\myfunc{$\textit{execute\_close}((a,b), (\balance_a, \balance_b))$ }{
$A(a) \gets A(a) + \balance_a$ \;
$A(b) \gets A(b) + \balance_b$ \;
$B(a,b) \gets \bot$ \;
\Return \textit{success}
}
\end{multicols}
\smallskip
\end{specification}

The state of a payment channel $B(a,b)$ is ${\{\mathbb{R}_{\geq 0} \times \mathbb{R}_{\geq 0} \} \cup \{\bot \}}$.
Intuitively, a unidirectional payment channel $B(a,b)$ can either be open, and then $B(a,b) = (\balance_{a}, \balance_{b})$, or closed, and then $B(a,b) = \bot$.
The initial state is that all unidirectional payment channels are closed, e.g., for payment channel $(a,b)$, the state is $ B(a,b) =  \bot$ at the beginning of the run.

The set of operations is the following:
\begin{itemize}
\item $\textit{open}((a,b),\amount)$.
A process that owns account $a$ can open a unidirectional payment channel with any other account $b$ with amount $\amount$, as long as it has enough balance in $A$ and does not already have an open payment channel with $b$.
The call returns \textit{success} if the channel is opened successfully and \textit{fail} otherwise.

\item $\textit{transfer}((a,b),\amount)$.
A payment in the payment channel $(a,b)$ is possible if the channel is open, if the caller of the operation $a$ is $\textit{owner}(a)$ and $a$ has enough balance in the channel to make the payment.
This call does not return a response.

\enlargethispage{1.9\baselineskip}
\item $\textit{source\_close}((a,b), \balance_a)$.
A source closing of a payment channel $(a,b)$ can be called by $\textit{owner}(a)$.
The call succeeds if the process that invokes the call does not try to close it with balance $\balance_a$ that is not the amount it has in the channel.
After the call ends, the balances in the channel are transferred to accounts $a,b$ in the asset transfer system.
The call returns \textit{success} if the channel is closed successfully and \textit{fail} otherwise.

\item $\textit{target\_close}((a,b), \balance_b)$.
A target closure of a payment channel $(a,b)$ is symmetrical to the source\_close call, but is invoked by $\textit{owner}(b)$.
\end{itemize}

We differentiate between two types of unidirectional payment channels, depending on whether the source close call is included in the allowed set of operations of the shared object or not.
Note that without source close, the source depends on the target to close the channel to receive its deposit back after the channel is opened.
However, for the target to receive its balance from the channel in the asset transfer system, it has to eventually close the channel.
When the target closes the channel, the source also receives its respective balance.

We do not consider a channel with source close and without target close, as in this case only the source has operations it can call, and the target relies on the source for all its operations regarding the channel, and cannot receive the funds transferred to it in the channel on-chain unless the source closes the channel.
Also, since in this case, only the source has operations it can invoke, this abstraction can be implemented easily in an asynchronous network as it does not require consensus or any interaction at all between the source and target.
We also believe this case does not correspond correctly to existing implementations of payment channels as discussed in \Cref{sec:relatedWork}.

\subsection{Impossibility of a unidirectional payment channel with source close}

We show that a unidirectional payment channel that has the source close operation has a consensus number of at least 2, and therefore cannot be implemented in an asynchronous message passing network.
Formally, we prove:

\begin{restatable}{lemma}{uniChannelWithSouceCloseLemma}
Consensus has a wait-free implementation for $2$ processes in the read-write shared memory model with an instance of a unidirectional payment channel with source close shared memory object and shared registers.
\end{restatable}

\begin{proof}
In \Cref{alg:reductionConsensuswithBidirectionalChannel1}, only process $p_1$ transfers money to $p_2$ using the channel, and they both attempt to close the channel: $p_1$ after the payment is made, and $p_2$ without accepting the payment.
Thus, changing the close call in \Cref{alg:reductionConsensuswithBidirectionalChannel1:p1close} to $B.\textit{source\_close}$ and the call in \Cref{alg:reductionConsensuswithBidirectionalChannel1:p2close} to $B.\textit{target\_close}$ yields a consensus algorithm among 2 processes using a unidirectional payment channel with source and target close operations.
\end{proof}

Based on the above lemma and FLP~\cite{fischer1985FLP}, we get the following result:
\begin{theorem} \label{cor:unidirectionalAsync}
There does not exist an implementation of the unidirectional payment channel abstraction with source close in the asynchronous message-passing model.
\end{theorem}

\subsection{Unidirectional payment channel without source close}
\label{sec:uniChanelImplementation:UpperBound}
Next, we discuss unidirectional payment channel without the source close operation.
We first prove a lower bound on the message complexity of an implementation and then prove that this lower bound is tight by providing an implementation for the abstraction in the asynchronous message-passing model.

\subparagraph{Lower bound.}
We prove that any algorithm that implements the unidirectional payment channel specification incurs a combined message complexity for open, transfer, and close of $\Omega(n^2)$.
To this end, we prove the following lemma:

\begin{restatable}{lemma}{uniChannelMessageComplexity}
\label{lem:uniChannelMessageComplexity}
Consider an algorithm that implements the unidirectional payment channel abstraction $B$, and an asset transfer $A$.
Then there exists a run with $B.\textit{open}$, $B.\textit{transfer}$, $B.\textit{target\_close}$, and $A.\textit{read}$ calls, in which correct processes send at least $(f/2)^2$ messages.
\end{restatable}

\begin{proof}
	We can simulate an $A.\textit{transfer}$ call between two accounts $a, b$ with initial balances $1, 0$, respectively.
	This is done by having $\textit{owner}(a)$ call $B.\textit{open}((a,b),1)$, followed $B.\textit{transfer}((a,b),1)$, and lastly, having $\textit{owner}(b)$ call $B.\textit{target\_close}((a,b),1)$.
	If accounts $a,b$ are owned by the same process $p$, we can construct exactly the the same runs used in the proof of \Cref{thm:assetTransferMsgComplexity} in order to prove this lemma.
	The only change is that we replace the $A.\textit{transfer}$ call invoked by $p$ in the original proof with the three calls mentioned above.
\end{proof}

\begin{algorithm*}[t!]\smaller[2]
\caption{\textbf{Unidirectional payment channel without source\_close implementation in the asynchronous message-passing model}. Operations for process $p$.}

\SetAlgoNoEnd
\label{alg:uniPaymentChannelImplementation}
\DontPrintSemicolon
\SetInd{0.2em}{0.2em}

\let\oldnl\nl
\newcommand{\nonl}{
\renewcommand{\nl}{\let\nl\oldnl}}

\SetKwBlock{sharedObjects}{Shared Objects:}{}
\SetKwBlock{localVariables}{Local variables:}{}
\SetKwProg{myproc}{Procedure}{:}{}
\SetKwProg{myfunc}{Function}{:}{}
\SetKwProg{myupon}{Upon}{:}{}

\nonl \sharedObjects{
\nonl	$A$ - asset transfer object \;
}
\BlankLine
\nonl\localVariables {
\nonl	$\textit{source}[] $ - a dictionary with the balances of all  channels that $p$ is the source, initially $\bot$  \;
\nonl	$\textit{target}[] $ - a dictionary with $A$ multisig invocations for all channels $p$ is the target, initially $\bot$ \;
}
\nonl\begin{minipage}{0.454\textwidth}%
\tcp{\underline{This call can be invoked by $\textit{owner}(a)$}}
\myproc{$\textit{open}((a,b),\amount)$ \label{alg:uniPaymentChannelImplementation:openCall}} {
\If{\smaller [0.5]$\textit{owner}(a) \neq p \vee \textit{source}[ab] \neq \bot  \vee  A.\textit{read}(a) < \amount$ \label{alg:uniPaymentChannelImplementation:openCallIf}} {
\Return \textit{fail}
}
invoke $A.\textit{transfer}(a, [(ab, \amount)]) $ \label{alg:uniPaymentChannelImplementation:openCallATransfer} \tcp*{$ab$ is multisig }
create $A.\textit{transfer}(ab, [(a, \amount), (b, 0)])$ \newline invocation $tx$ \label{alg:uniPaymentChannelImplementation:createOpenTX}\;
add $p$'s signature to $tx$ \tcp*{$tx$ is not a valid\phantom{zzzzzzz} transaction without $\textit{owner}(b)$'s signature}\label{alg:uniPaymentChannelImplementation:sendOpenTX}
\textbf{send} $\langle \text{"open"}, tx, \amount \rangle$ to $\textit{owner}(b)$ \;
$\textit{source}[ab] \gets (\amount, 0)$ \;
\Return \textit{success}
}
\end{minipage}\phantom{XXX}%
\begin{minipage}{0.45\textwidth}
\tcp{\underline{This message is received by $\textit{owner}(b)$}}
\nl\myupon{receiving $\langle \text{"open"}, tx, \amount \rangle$ and $A.\textit{read}(ab) = \amount$\label{alg:uniPaymentChannelImplementation:receiveOpenMsg}} {
let $tx$ be $A.\textit{transfer}(ab, [(a, \amount), (b, 0)])$ invocation\;
\If{$\textit{owner}(b)\neq p \vee\neg \textit{validate}(tx, a) \vee \textit{target}[ab] \neq \bot$ } {
\Return 
}
$\textit{target}[ab] \gets tx$ \;
}
\vspace{4.8em}
\end{minipage}
\BlankLine 
\nonl\begin{minipage}{0.46\textwidth}%
\tcp{\underline{This call can be invoked by $\textit{owner}(a)$}}
\myproc{$\textit{transfer}((a,b),\amount)$\label{alg:uniPaymentChannelImplementation:transferCall}} {
\If{$\textit{owner}(a) \neq p \vee \textit{source}[ab] = \bot$\label{alg:uniPaymentChannelImplementation:transferCall:if1}}
{
\Return
}
$(\balance_a, \balance_b) \gets \textit{source}[ab]$ \;
\If{$\balance_a < \amount $ \label{alg:uniPaymentChannelImplementation:transferCall:if2}}
{
\Return 
}
$(\textit{new\_bal}_a, \textit{new\_bal}_b) \gets (\balance_a - \amount, \balance_b + \amount)$ \;
{\smaller[0.5] create $A.\textit{transfer}(ab, [(a, \textit{new\_bal}_a), (b,\textit{new\_bal}_b)])$ invocation $tx$} \label{alg:uniPaymentChannelImplementation:createTransferTX}\;
add $p$'s signature to $tx$ \label{alg:uniPaymentChannelImplementation:addSignatureToTransfer}\;
\textbf{send} $\langle \text{"transfer"}, tx, \amount \rangle$ to $\textit{owner}(b)$ \label{alg:uniPaymentChannelImplementation:transferCallSendTX}\;
$\textit{source}[ab] \gets (\textit{new\_bal}_a, \textit{new\_bal}_b)$ \;

}
\end{minipage}\phantom{XXk}%
\begin{minipage}{0.46\textwidth}
\tcp{\underline{This message is received by $\textit{owner}(b)$}}
\nl\myupon{receiving $\langle \text{"transfer"}, tx, \amount \rangle$} {
let $tx$ be $A.\textit{transfer}(ab, [(a, \balance_a), (b, \balance_b)])$  invocation 

\If{$\textit{owner}(b) \neq p \vee \neg \textit{validate}(tx, a) \vee \textit{target}[ab] = \bot$ \label{alg:uniPaymentChannelImplementation:transferCallMsg}}
{
\Return 
}
get $A.\textit{transfer}(ab, [(a, \textit{c\_bal}_a), (b, \textit{c\_bal}_b)])$ invocation from $\textit{target}[ab]$ \tcp*{The currently\phantom{zzzz} \phantom{zzzzzzzzzzzzzzzzzzzzzzzzzzzzzz}stored transaction}
\If{$\textit{c\_bal}_a \neq \balance_a - \amount \vee \textit{c\_bal}_b \neq \balance_b + \amount$} {
\Return 
}
$\textit{target}[ab] \gets tx$ \tcp*{store new transaction}
}	
\end{minipage}
\BlankLine 
\nonl\begin{minipage}{0.45\textwidth}%
\tcp{\underline{This call can be invoked by $\textit{owner}(b)$}}
\myproc{$\textit{target\_close}((a,b),\balance_b)$\label{alg:uniPaymentChannelImplementation:closeCall}} {
\If{$\textit{owner}(b) \neq p \vee \textit{target}[ab] = \bot $ }
{
\Return \textit{fail}
}
get $A.\textit{transfer}(ab, [(a, \textit{curr\_bal}_a), (b, \textit{curr\_bal}_b)])$ transaction $tx$ from $\textit{target}[ab]$ \;
\If{$\balance_b \neq \textit{curr\_bal}_b$} {
\Return \textit{fail}
}
add $p$'s signature to $tx$  \tcp*{complete the multisig} \label{alg:uniPaymentChannelImplementation:addSignatureToClose} 
invoke $tx$ \tcp*{invoke $A$ with closing transaction}  \label{alg:uniPaymentChannelImplementation:closeCallInvokeA} 
$\textit{target}[ab] \gets \bot$  \;
\textbf{send} $\langle \text{"close"}, (a,b) \rangle $ to $\textit{owner}(a)$\;
\Return \textit{success}
}
\end{minipage}\phantom{XXXl}%
\begin{minipage}{0.46\textwidth}
\tcp{\underline{This message is received by $\textit{owner}(a)$}}
\nl\myupon{receiving $\langle \text{"close"}, (a,b)\rangle$ and $A.\textit{read}(ab) = 0$\label{alg:uniPaymentChannelImplementation:receiveCloseMsg}} { 
$\textit{source}[ab] \gets \bot$ \label{alg:uniPaymentChannelImplementation:closeSourceBot}\;
}
\BlankLine\BlankLine
\BlankLine  \BlankLine 
\BlankLine  \BlankLine 
\myfunc{$\textit{validate}(tx, a)$} {
\Return $tx$ is a valid invocation of $A$ and it contains $\textit{owner}(a)$'s signature \;
}
\vspace{4.8em}
\end{minipage}
\end{algorithm*}

\subparagraph{Upper bound.} 
We provide an algorithm in the asynchronous message-passing model that implements a unidirectional payment channel without source close.
The algorithm assumes an asset transfer system $A$, implemented as in~\cite{guerraoui2019consensus}, and discussed in~\Cref{sec:assetTransfer:upperBound}.
The algorithm is detailed in \Cref{alg:uniPaymentChannelImplementation}.
We denote an account name with a string $c_1c_2\cdots c_k \in A$ to refer to an account with a public key that is a $k$-of-$k$ multisignature of $\{\textit{owner}(c_1), \ldots, \textit{owner}(c_k)\}$.
For example, to sign an invocation of the asset transfer object $A$ of account $ab$, like transferring money from $ab$ to another account, both 
$\textit{owner}(a)$ and $\textit{owner}(b)$  need to sign the message with their respective private keys before the call can be invoked.
The transfer call with an appropriate multisignature can be invoked by any process, in particular, the last process to sign the invocation and complete the signature.
When the algorithm mentions that a process creates an $A.\textit{transfer}$ invocation, e.g., in lines~\ref{alg:uniPaymentChannelImplementation:createOpenTX} and~\ref{alg:uniPaymentChannelImplementation:createTransferTX}, it does not mean the process invokes the transfer operation, but rather that it adds its signature to a multisignature message allowing an invocation of $A$'s operation.
Any invocation of $A.\textit{transfer}$ call is explicitly mentioned (lines~\ref{alg:uniPaymentChannelImplementation:openCallATransfer}, \ref{alg:uniPaymentChannelImplementation:addSignatureToClose}).
We further assume FIFO order on messages sent between every two processes. This can be easily implemented with sequence numbers.
\enlargethispage{-\baselineskip}

We explain below the implementation details of the algorithm for each of the operations:
\begin{itemize}
\item Open.
The open procedure of a channel $(a,b)$ (\Cref{alg:uniPaymentChannelImplementation:openCall}) requires $p_1 =  \textit{owner}(a)$ to make an initial deposit by invoking the transfer method of $A$ from account $a$ to a multisignature account $ab$ (\Cref{alg:uniPaymentChannelImplementation:openCallATransfer}).
After the transfer is completed, $p_1$ creates a transaction $tx$ that transfers the deposit back to its account and $0$ to $p_2 = \textit{owner}(b)$ and sends it to $p_2$ (\Cref{alg:uniPaymentChannelImplementation:sendOpenTX}).
Note that at this stage, process $p_1$ cannot invoke $A$ with $tx$ since it requires a multisignature, but when $p_2$ receives it, it can add its signature as well and then invoke $A$ with the $tx$.

When $p_2$ receives $tx$, this transaction message it also waits for the balance in account $ab$ to reflect the deposit (\Cref{alg:uniPaymentChannelImplementation:receiveOpenMsg}) to ensure the money was deposited in account $ab$ using the asset transfer system, after which it considers the account as open.

\item Transfer.
When $p_1$ wants to transfer money in an open channel $(a,b)$ (\Cref{alg:uniPaymentChannelImplementation:transferCall}) it creates a transaction $tx$ which is an $A.\textit{transfer}$ invocation transferring money from the multisignature account $ab$ to accounts $a$ and $b$ with the last balance of the channel after the payment.
E.g., if the balance of the channel is $(10,1)$, and $p_1$ wants to make a payment of $1$ on the channel, it creates transaction $tx$ required to invoke $A.\textit{transfer}(ab,[(a,9),(b,2)])$, which transfers $9$ money units to $p_1$ and $2$ to $p_2$.
Then $p_1$ adds its signature to $tx$ (\Cref{alg:uniPaymentChannelImplementation:addSignatureToTransfer}), and sends it to $p_2$, which stores it.
Note that $p_1$ cannot invoke $A$ with $tx$ since it is still missing $p_2$'s signature.
Thus, making a payment on the channel simply requires one message from the source user to the target user containing $tx$, and multiple payments can be made on the channel without invoking $A$'s transfer call.

\item Close.
When $p_2$ wants to close the channel $(a,b)$ (\Cref{alg:uniPaymentChannelImplementation:closeCall}), it takes the last transaction of account $ab$ it received from $p_1$ and adds its signature to it (\Cref{alg:uniPaymentChannelImplementation:addSignatureToClose}).
$p_2$'s signature completes the multisignature, making it a valid transaction, and allowing $p_2$ to use it to invoke $A$'s transfer operation (\Cref{alg:uniPaymentChannelImplementation:closeCallInvokeA}).
Process $p_2$ notifies $p_1$ that it closed the channel, after which $p_1$ considers the channel closed.
After the channel is closed, $p_1$ can  reopen it with a new call of the open operation.
\end{itemize}
Thus, opening and closing of the channel requires invoking a single $A.\textit{transfer}$ operation, which incurs $O(n^2)$ messages because of the broadcast, but transferring money on the channel itself requires only one message per transfer. 

\textbf{Correctness.}
We prove below that the implementation (\Cref{alg:uniPaymentChannelImplementation}) is Byzantine sequentially consistent (BSC) with respect to the sequential specification (Specification~\ref{alg:uniPaymentChannelOps}).

\begin{definition} \label{def:historyAugmentation}
Let $E$ be an execution of \Cref{alg:uniPaymentChannelImplementation} and $H$ its matching history.
Let $\widetilde{H}$ be a completion of $H$ by removing any pending open and close calls that did not reach $A$'s transfer call invocation (Lines \ref{alg:uniPaymentChannelImplementation:openCallATransfer} and \ref{alg:uniPaymentChannelImplementation:closeCallInvokeA}, respectively), and let $\widetilde{H}|_c$ be $\widetilde{H}$'s history with the operations of correct processes.

Define $H'$ as an augmentation of $\widetilde{H}|_c$ as follows:
For any correct process $q = \textit{owner}(b)$ that invokes a successful $B.\textit{target\_close}((a,b), \balance_b)$ s.t. process $p = \textit{owner}(a)$ is a Byzantine process, we add before the target close call the following two invocations to $H'$ by $p$:
\begin{itemize}
\item $B.\textit{open}((a,b), \balance_b)$ with an account $a$ s.t. $A(a) \geq \balance_b$.
Since account $a$ has enough money to open the channel, the open call succeeds.

\item $B.\textit{transfer}((a,b), \balance_b)$ which is invoked immediately after the previous open call returns.
\end{itemize}
\end{definition}

The two added Byzantine calls ensure that when $q$ invokes the close operation, it succeeds.
Next, we construct a linearization of $H'$.
\enlargethispage{1.6\baselineskip}

\begin{definition}
Let $H'$ be the augmented history of \Cref{alg:uniPaymentChannelImplementation} as defined in \Cref{def:historyAugmentation}.
Let $E'$ be a linearization of $H'$ by defining the following linearization points:
\begin{itemize}
\item Any open or close call that fails is linearized immediately after its invocation.

\item A transfer call that returns because of the if statements (Lines \ref{alg:uniPaymentChannelImplementation:transferCall:if1}, \ref{alg:uniPaymentChannelImplementation:transferCall:if2}) is linearized immediately after its invocation.

\item Any successful $\textit{open}((a,b), \amount)$ s.t. $q = \textit{owner}(b)$ is a correct process, then it linearizes after $q$ reaches \Cref{alg:uniPaymentChannelImplementation:receiveOpenMsg}.
If $q$ is Byzantine, the call linearizes when it ends.

\item Any $\textit{transfer}((a,b), \amount)$ that reaches \Cref{alg:uniPaymentChannelImplementation:transferCallSendTX} s.t. $q = \textit{owner}(b)$ is a correct process, then it linearizes after $q$ reaches \Cref{alg:uniPaymentChannelImplementation:transferCallMsg}.
If $q$ is Byzantine, the call linearizes when it ends.

\item Any successful $\textit{target\_close}((a,b),  \balance_b)$ s.t. $p = \textit{owner}(a)$ is a correct process, then it linearizes after $p$ reaches \Cref{alg:uniPaymentChannelImplementation:receiveCloseMsg}.
If $p$ is Byzantine, the call linearizes when it ends.
\end{itemize}
\end{definition}

The open and transfer calls change the state of the channel.
By the sequential specification, the target can call target close with this new state.
Therefore, the linearization point of these calls occur after the target receives the message informing it of the new state.
Regarding the close call linearization point: the source can only reopen a channel after it learns that the channel has been closed, and therefore the linearization point is when the source receives the information of the closure and verifies it on-chain.

Next, we provide below the lemmas showing that the linearization $E'$ satisfies the sequential specification.

\begin{restatable}{lemma}{openLemma}
A $B.\textit{open}$ call for channel $(a,b)$ succeeds only if the channel is closed when the call is invoked.
\end{restatable}
\begin{proof}
	Immediate from the algorithm.
	A channel $(a,b)$ opening fails if $\textit{source}[ab] = \bot$ (\Cref{alg:uniPaymentChannelImplementation:openCallIf}).
	This is the case for all channels at the beginning of the run, or if the channel was previously closed successfully (\Cref{alg:uniPaymentChannelImplementation:closeSourceBot}).
\end{proof}

\begin{restatable}{lemma}{transferLemma}
For any $B.\textit{transfer}$ call for channel $(a,b)$ there is a preceding open call for the channel in $H'$.
\end{restatable}
\begin{proof}
	If the transfer call in $H'$ is invoked by a correct process, then from the algorithm $\textit{source}[ab] \neq \bot$.
	This is only possible by the algorithm if $p$ invokes an open call for the channel before the transfer invocation.
	If $q$ is correct, then the open and transfer calls linearize when $q$ receives the messages for the corresponding calls (Lines \ref{alg:uniPaymentChannelImplementation:receiveOpenMsg}, \ref{alg:uniPaymentChannelImplementation:transferCallSendTX}).
	Since we assume FIFO order on the links between any two processes, then the open call linearizes before the transfer call.
	If $q$ is Byzantine, then the open and transfer calls linearize immediately after they successfully return.
	
	If $p$ is Byzantine, then the transfer invocation is in $H'$ because there is some close invocation by a correct process $q$.
	Before that, there is also a matching open call by $p$.
	The open call linearizes before the transfer call.
\end{proof}

\begin{restatable}{lemma}{closeLemma}
For any successful $B.\textit{target\_close}((a,b), \balance_b)$ call in $H'$ there is a preceding $B.\textit{open}((a,b), \amount)$ call for channel $(a,b)$ s.t. $\amount \geq \balance_b$, followed by a $B.\textit{transfer}$ call that changes the state of the channel to $(\balance_a, \balance_b)$ for $\balance_{a} = \amount - \balance_b$.
\end{restatable}
\begin{proof}
	If the close call ends successfully, then $q$ has in $\textit{target}[ab]$ a valid transaction $A.\textit{transfer}(ab,[(a, \balance_a), (b, \balance_b)])$, otherwise, the call fails.
	Therefore, if $p$ is a correct process, it opens the channel $(a,b)$ with some deposit \amount, and transfers in the channel s.t. the balances in the channel change to $(\balance_a, \balance_b)$.
	Both the open and transfer are linearized before the close invocation, otherwise, the close call fails.
	If $p$ is Byzantine, then we add the matching open and transfer invocations to $H'$ which are linearized before the close invocation.
\end{proof}

\begin{restatable}{lemma}{eventuallyEndsLemma}
In an infinite execution of \Cref{alg:uniPaymentChannelImplementation} every call invoked by a correct process eventually returns.
\end{restatable}
\begin{proof}
	In all cases where an open or close calls return \textit{fail} it does so immediately, since it is done prior to any invocation of $A$.
	Transfer calls that return due to the if statements (Lines \ref{alg:uniPaymentChannelImplementation:transferCall:if1}, \ref{alg:uniPaymentChannelImplementation:transferCall:if2}) also return immediately.
	
	For the open call, the if condition in the channel open call (\Cref{alg:uniPaymentChannelImplementation:openCallIf}) ensures that the process that invokes the call owns account $a$ and that it has enough balance to open the channel.
	We also assume that a correct process does not invoke a new call before a previous call has a response event.
	Therefore, the conditions checked during the if statement hold when $A$'s transfer call is invoked, and by the asset transfer specification the call succeeds.
	
	A transfer call that does not invoke any of $A$'s calls, nor does it wait for a reply after it sends the transaction in \Cref{alg:uniPaymentChannelImplementation:transferCallSendTX}.
	Therefore, this call also returns immediately.
	
	A target\_close call that returns success invokes $A$ with a transfer call that transfers money from account $ab$ (\Cref{alg:uniPaymentChannelImplementation:closeCallInvokeA}).
	To reach this line, the process has to first check if the channel is open, and it has the matching transaction in $\textit{target}[ab]$ during the if statement of the call.
	Therefore, invoking $A$ will eventually succeed by the asset transfer specification, and the target\_close call returns successfully.
\end{proof}

Thus, we can conclude the following result from the lemmas above:
\begin{theorem}
\Cref{alg:uniPaymentChannelImplementation} implements a Byzantine sequentially consistent unidirectional payment channel without source close abstraction.
\end{theorem}
\subparagraph*{Changing the algorithm to be Byzantine linearizable.}
The algorithm can be changed to be Byzantine linearizable by having each message answered with an ack message.
E.g., after the open message is received in \Cref{alg:uniPaymentChannelImplementation:transferCall:if1}, the process sends an ack message back to the original sender.
The linearization point then is when the ack message is received.
This change requires sending more messages as part of the algorithm and also extends the latency, but this change does not affect the overall asymptotic message complexity.
This is also the reason why we choose BSC as the correctness criterion, not Byzantine linearizability. {}

\enlargethispage{1.5\baselineskip}
\section{Chain payments}
\label{sec:chain}
We can extend the discussion of payment channels to chain payments.
A chain payment system allows making payments off-chain between users who do not share a direct payment channel between them but do share a route through intermediate users.
For example, suppose that Alice wants to make a payment to Bob, but she does not share a direct payment channel with him.
Rather, she has a channel with Charlie and Charlie has a channel with Bob.
A chain payment allows using the route from Alice to Bob via Charlie to make the payment on all channels atomically.
Chain payments are used extensively in Lightning Network~\cite{poon2016lightning}.

The intuitive way to define a chain payment abstraction is with a single operation that makes a payment through a chain $((a_1, a_2), (a_2, a_3), \ldots, (a_{k-1}, a_k))$, and the outcome of the payment affects all channels of the chain in an atomic manner.
For example, suppose that the balances of the above channels of the chain are $(\balance_1, \balance_2), \ldots, (\balance_{k-1}, \balance_k)$, respectively, and a payment of \amount is made via the chain.
Then, after the linearization point, the balances of the channels are $(\balance_1-\amount, \balance_2+\amount), \ldots, (\balance_{k-1}-\amount, \balance_k+\amount)$, respectively.
In this case, assuming that the channels are bidirectional or unidirectional with source close, it can be proven that the consensus number~\cite{herlihy1991wait} of such chain payment object is $k$, meaning this object can be used to solve consensus between $k$ processes, in a similar manner to the $2$-consensus we prove in this paper for these channel types.

We also note that even if the channels of the chain are unidirectional without source close, implementing a chain payment is not intuitive and straightforward in asynchronous networks.
A possible solution is to adopt the use of \emph{Hash Time-Locked Contracts (HTLCs)}~\cite{decker2015fast,HTLCwiki} which are in use in the Lightning Network for chain payments.
An HTLC is a special conditioned payment between two users Alice and Bob.
An HTLC allows Alice to make a conditioned payment to Bob that includes some timeout $\Delta$ and a hash value $y$.
Bob can receive the payment if he exposes on-chain a value $x$ whose hash value is $y$ before $\Delta$ elapses.
We can equip the asset transfer system with a similar Hash Time Contract operation, without the timeout component.
That is, Bob receives the payment if he exposes $x$ without any timeout assumptions.
With this, we can implement chain payments in asynchronous networks based on the unidirectional channels without source close we presented in \Cref{alg:uniPaymentChannelImplementation}.
We leave as open future work to formalize these ideas and explore other possible asynchronous implementations of  chain payments. 

\section{Related work}
\label{sec:relatedWork}
Pedone and Schiper~\cite{pedone2002handling} identified that a transfer-like system can be implemented under asynchrony.
Gupta~\cite{gupta2016non}  showed that consensus is not needed for payment transactions.
The first definition of the asset transfer abstraction is due to Guerraoui et al.~\cite{guerraoui2019consensus}.
Subsequentially, Auvolat et al.~\cite{auvolat2020money} provide a weaker specification for asset transfer and detail an implementation that uses a broadcast service that guarantees FIFO order between every two processes.
Astro~\cite{collins2020online} implements and empirically evaluates an asynchronous payment system.

Solving Byzantine consensus in an asynchronous network is possible by using randomization to circumvent the FLP~\cite{fischer1985FLP} result.
Earlier protocols such as~\cite{ben1994power,rabin1983randomized} have exponential message complexity.
Later protocols such as~\cite{miller2016honey,abraham2019asymptotically,guo2020dumbo,keidar2021all} improve the message complexity in various settings, but they are not deterministic, and therefore their performance can only be measured in the expected case.

\enlargethispage{1.5\baselineskip}
There is also extensive research done on scaling cryptocurrencies using payment channels, including the Lightning Network~\cite{poon2016lightning}, Teechain~\cite{lind2019teechain}, Bolt~\cite{green2017bolt}, Sprites~\cite{miller2017sprites}, Perun~\cite{dziembowski2017perun}, Duplex Micropayment Channels~\cite{decker2015fast}, Raiden~\cite{raiden2020raidenImp}, and more.
In Ethereum~\cite{wood2014ethereum}, there are multiple second-layer networks like Arbitrum~\cite{kalodner2018arbitrum}, StarkNet~\cite{ben2018scalable}, and Optimism~\cite{optimismNetwork}.
All these works assume an underlying synchronous network to function correctly. 
E.g., Lightning Network uses a penalizing mechanism where one side of the channel can confiscate the other side's balance in the channel if it misbehaves and tries to close the channel at a stale state.
But for this mechanism to work correctly, the penalizing party has to place a transaction on-chain within a certain time period, thus requiring a synchronous network.
Brick~\cite{avarikioti2021brick} is a payment channel that preserves safety and liveness in asynchrony, but to do so requires a rather complex third-party entities (referred to as wardens) which validate channel transactions before they can be placed on-chain.
This work also assumes rational (and not Byzantine) behavior of the wardens.

Spilman~\cite{spilman_2013} proposed in 2013 a unidirectional payment channel implementation that shares similar concepts to \Cref{alg:uniPaymentChannelImplementation}.
Spilman's design has a timeout for each channel. 
The target of the channel has to close the channel before the timeout passes, otherwise, the source side can refund its initial deposit in the channel.
Because of the timeout, this design relies on network synchrony.
To the best of our knowledge, our implementation of a unidirectional channel without source close is the first that works on an underlying asynchronous network without requiring any third-party assistance to operate the channel, and in the presence of Byzantine processes.

We are also the first to show a quadratic lower bound for payments in asset transfer systems, as well as to explore second-layer payment channels as a scaling solution in asynchrony.

\section{Conclusion}
\label{sec:conclusion}
In this paper we presented the possibility of using payment channels in asynchronous asset transfer systems as a scaling solution.
We showed that an asset transfer system requires a quadratic message complexity per payment.
Then, we showed a series of possibility and impossibility results regarding payment channels as a scaling solution. {}

\bibliography{p029-Naor}

\end{document}